\documentclass[a4paper,pdftex,reqno,12pt]{amsart}

\allowdisplaybreaks[1]

\usepackage{geometry}

\usepackage{graphicx}
\usepackage{enumerate}
\usepackage{courier}
\usepackage{times}
\usepackage{amsmath,amssymb}

\usepackage{hyperref}

\theoremstyle{plain}
\newtheorem{Theorem}{Theorem}[section]
\newtheorem{Proposition}[Theorem]{Proposition}

\newtheorem{Corollary}[Theorem]{Corollary}
\newtheorem{Lemma}[Theorem]{Lemma}

\theoremstyle{definition}
\newtheorem{Definition}[Theorem]{Definition}

\theoremstyle{remark}

\begin{document}


\title{Bounds for the diameter of the weight polytope}

\date{}

\author{Sascha Kurz}
\address{Sascha Kurz, University of Bayreuth, 95440 Bayreuth, Germany}
\email{sascha.kurz@uni-bayreuth.de}

\begin{abstract}
  A weighted game or a threshold function in general admits different weighted representations even if the sum of non-negative weights is fixed to one. 
Here we study bounds for the diameter of the corresponding weight polytope. It turns out that the diameter can be upper bounded in terms 
of the maximum weight and the quota or threshold. We apply those results to approximation results between power distributions, given by power indices, 
and weights.
\end{abstract}

\maketitle

\section{Introduction}
\label{sec_introduction} 

Consider a stock corporation whose shares are hold by three major stockholders owning 
35\%, 34\%, and 17\%, respectively. The remaining 14\% are widely spread. 
Assuming that decisions are made by a simple majority rule, all three major stockholders have  
equal influence on the company's decisions, while the private shareholders have no say. To 
be more precise, any two major stockholders can adopt a proposal, while the private   
shareholders together with an arbitrary major stockholder need further affirmation. 
Such decision environments can be captured by means of weighted voting games. Formally, a 
weighted (voting) game consists of a set of players or voters $N=\{1,\dots,n\}$, a vector of non-negative weights 
$w=(w_1,\dots,w_n)$, and a positive quota $q$. A proposal is accepted if and only if the weight sum 
of its supporters meets or exceeds the quota. 

Committees that decide between two alternatives have received wide attention. Von Neumann and Morgenstern 
introduced the notion of simple  games, which is a super class of weighted games, in \cite{von1953theory}. Examples of decision-making bodies that 
can be modeled as weighted games are the US Electoral College, the Council of the European Union, the 
UN Security Council, the International Monetary Fund or the Governing Council of the European Central Bank.
Many applications seek to evaluate players' influence or power in simple or weighted games, see, e.g., 
\cite{leech1987ownership}. The initial example illustrates that shares or 
weights can be a poor proxy for the distribution of power. Using the taxicab metric, i.e., the 
$\Vert\cdot\Vert_1$-distance, the corresponding distance between shares and relative power is 
$\left|0.35-\tfrac{1}{3}\right|+\left|0.34-\tfrac{1}{3}\right|+\left|0.17-\tfrac{1}{3}\right|+
\left|0.14-0\right|\approx$32.67\%. If the weights add up to one, then we speak of relative or normalized weights.
The insight that the power distribution differs from relative weights, triggered the invention of so-called power 
indices like the Shapley-Shubik index \cite{shapley1954method}, the Penrose-Banzhaf index \cite{banzhaf1964weighted},
or the nucleolus \cite{schmeidler1969nucleolus}. Due to the combinatorial nature of most of those power 
indices, qualitative assessments are technically demanding and large numbers of involved parties cause computational 
challenges \cite{chalkiadakis2011computational}.

One reason for the difference between relative weights and power is that a weighted game permits different 
representations. If there are two normalized representations whose weight vectors are at large distance then at least 
one of the relative weight vectors also has a large distance to the power distribution. So, here we study bounds for the 
diameter of the weight polytope, i.e., bounds for the maximal distance between two normalized vectors of  
the same weighted game. We will study those bounds in terms of the number of players, the relative quota, and 
the maximum relative weight in a given representation of the game.

Each weighted game, also called threshold function in threshold logic, admits a representation with integer weights. Bounds for the necessary 
magnitude of integer weighted are studied in the literature, see e.g.\ \cite{babai2010weights} and the references therein.

The remaining part of the paper is structured as follows. In Section~\ref{sec_weight_polytope} we give the necessary definitions 
for simple games, weighted games and the weight polytope. Worst case lower bounds on the diameter of the weight polytope are 
given in Section~\ref{sec_lower_bounds} and upper bounds are given in Section~\ref{sec_upper_bounds}. Applications to 
approximation results for power indices are given in Section~\ref{sec_applications} before we draw a brief conclusion 
in Section~\ref{sec_conclusion}. Some lengthy or more technical proofs are moved to an appendix.

\section{The weight polytope of a weighted game}
\label{sec_weight_polytope}

For a positive integer $n$ let $N=\{1,\dots, n\}$ be the set of players. A \emph{simple game} is a 
mapping $v\colon 2^N\to\{0,1\}$ from the subsets of $N$ to binary outcomes satisfying 
$v(\emptyset)=0$, $v(N)=1$, and $v(S)\le v(T)$ for all $\emptyset\subseteq S\subseteq T\subseteq N$. 
The interpretation in the context of binary voting systems is as follows. A subset $S\subseteq N$, also called coalition,  
is considered as the set of players that are in favor of a proposal, i.e., which vote {\lq\lq}yes{\rq\rq}.  
If $v(S)=1$ we call coalition $S$ winning and losing otherwise. By $\mathcal{W}(v)$ we denote the set of winning 
coalitions and by $\mathcal{L}(v)$ we denote the set of losing coalitions of $v$. If coalition $S$ is winning but each proper subset is 
losing, then we call $S$ minimal winning. Similarly, if $S$ is losing but each proper superset of $S$ is winning, then 
we call $S$ maximal losing. By $\mathcal{W}^m(v)$ we denote the set of minimal winning and by $\mathcal{L}^m(v)$ we denote the 
set of maximal losing coalitions. $v(S)$ encodes the group decision, i.e., 
$v(S)=1$ if the proposal is accepted and $v(S)=0$ otherwise. So, these assumptions for a simple game 
are quite natural for a voting system with binary options in the input and output domain. The dual $v^d$ of a simple game $v$ 
is defined via $v^d(S)=v(N)-v(N\backslash S)=1-v(N\backslash S)$ for all $S\subseteq N$ and is a simple game itself. If $v(S)=v(S\cup\{i\})$ for all 
$S\subseteq N$, then we call player~$i$ a null player. Player~$i$ is a passer if $v(\{i\})=1$. Two players $i$ and $j$ are 
equivalent if $v(S\cup\{i\})=v(S\cup\{j\})$ for all $S\subseteq N\backslash\{i,j\}$.

A simple game $v$ is called \emph{weighted} if there exist weights $w\in\mathbb{R}_{\ge 0}^n$ 
and a quota $q\in\mathbb{R}_{>0}$ such that $v(S)=1$ if and only if $w(S):=\sum_{i\in S} w_i\ge q$. From the conditions 
of a simple game we conclude $0<q\le w(N)$. If $w(N)=1$ we speak of normalized or relative weights, where 
$0<q\le 1$. We denote the respective game by $v=[q;w]$ and refer to the pair $(q;w)$ as a weighted 
representation, i.e., we can have $[q;w]=[q';w']$ but $(q;w)\neq (q';w')$. The example from the introduction 
can, e.g., be represented by $(51\%;35\%,34\%,17\%,14\%)$, $\left(\tfrac{1}{2};\tfrac{1}{3},\tfrac{1}{3},\tfrac{1}{3},0\right)$, 
or $(6;4,3,3,1)$, where the fourth player mimics the private shareholders. 
\begin{Lemma}
  \label{lemma_dual_weights}
  If $(q;w)$ is a normalized representation of a weighted game $v$, then $(1-q+\varepsilon;w)$ is a normalized representation of the 
  dual game $v^d$ for each $0<\varepsilon< \min\{q-w(S)\mid S\in\mathcal{L}(v)\}$.
\end{Lemma}
\begin{proof}
  For each losing coalition $S$ of $v^d$ the coalition $N\backslash S$ is winning in $v$, so that $w(N\backslash S)=1-w(S)\ge q$ and 
  $w(S)\le 1-q<1-q+\varepsilon$. Now let $S$ be a winning coalition of $v^d$, so that $N\backslash S$ is losing in $v$ and $\varepsilon<q-w(N\backslash S)=q-1+w(S)$, 
  which is equivalent to $w(S)>1-q+\varepsilon$. Since $\emptyset$ is a losing coalition in $v$ we have $\varepsilon<q-w(\emptyset)=q$, so 
  that $1-q+\varepsilon<1$.
\end{proof}
Note that $\min\{q-w(S)\mid S\in\mathcal{L}(v)\}>0$.

Given a weighted game $v$, we call a weight vector $w\in\mathbb{R}_{\ge 0}$ feasible if there exists a quota $q\in\mathbb{R}_{>0}$ 
satisfying $v=[q;w]$. Obviously, such a quota exists iff the largest weight of a losing coalition is strictly smaller than the smallest 
weight of a winning coalition. Thus, c.f.~\cite[Lemma 3.2]{kaniovski2015representation}, the set of feasible normalized weight vectors is given by
\begin{eqnarray*}
  &&\left\{ w\in\mathbb{R}_{\ge 0}^n\mid w(N)=1, v(S)>v(T)\quad\forall S\in\mathcal{W}(v), T\in\mathcal{L}(v)\right\}\\
  &=& \left\{ w\in\mathbb{R}_{\ge 0}^n \mid w(N)=1, v(S)>v(T)\quad\forall S\in\mathcal{W}^m(v), T\in\mathcal{L}^m(v)\right\}.
\end{eqnarray*}
Note that these sets only depend on the game $v$ and are non-empty for weighted games. Due to the involved strict inequalities 
we have to consider their closure in order to obtain polytopes.

\begin{Definition}
  \label{def_weighted_polytope}
  For a weighted game $v$ we define the weight polytope of $v$ by
  $$
    \mathsf{W}(v)= \left\{ w\in\mathbb{R}_{\ge 0}^n \mid w(N)=1, v(S)\ge v(T)\quad\forall S\in\mathcal{W}^m(v), T\in\mathcal{L}^m(v)\right\}
  $$
  and call 
  $$
    \operatorname{diam}(\mathsf{W}(v))=\max\left\{\Vert w-w'\Vert_1 \mid w,w'\in \mathsf{W}(v)\right\}
  $$
  its diameter, where $\Vert x\Vert_1:=\sum_i \left| x_i\right|$. 
\end{Definition}
As an example we consider the weighted game $v=[2;1,1,1]$. For $w\in\mathsf{W}(v)$ the conditions  $w(S)\ge w(T)$ for all $S\in\mathcal{W}^m(v)$ and all 
$T\in\mathcal{L}^m(v)$ read $w_1+w_2\ge w_3$, $w_1+w_3\ge w_2$, and $w_2+w_3\ge w_1$. The normalization $w(N)=1$ can be used to eliminated $w_3$ via 
$w_3=1-w_1-w_2$. Finally, respecting $w\in\mathbf{R}_{\ge 0}^3$ gives
$$
  \mathsf{W}(v)=\left\{ \left(w_1,w_2,1-w_1-w_2\right)\mid  0\le w_1\le \frac{1}{2}, 0\le w_2\le \frac{1}{2}, w_1+w_2\ge \frac{1}{2}\right\}.
$$
Since $w:=\left(\tfrac{1}{2},\tfrac{1}{2},0\right)\in\mathsf{W}(v)$ and $w':=\left(\tfrac{1}{2},0,\tfrac{1}{2}\right)\in\mathsf{W}(v)$, we have 
$$
  \operatorname{diam}(\mathsf{W}(v))\ge \Vert w-w'\Vert_1=1. 
$$  
Indeed, it can be shown that $\Vert \tilde{w}-\hat{w}\Vert_1\le 1$ for all $\tilde{w},\hat{w}\in \mathsf{W}(v)$, so that $\operatorname{diam}(\mathsf{W}(v))=1$ in 
our example.  

For a simple game $v$ the set $\mathsf{W}(v)$ is non-empty iff $v$ is a so-called roughly weighted game, which is a relaxation 
of a weighted game. While also for a weighted game $v$ not any element in $\mathsf{W}(v)$ can be completed by a suitable quota $q\in(0,1]$ to a normalized representation $(q;w)$, 
Definition~\ref{def_weighted_polytope} makes sense nevertheless since $\dim(\mathsf{W}(v))=n-1$, see 
e.g.~\cite[Lemma 3.4]{kaniovski2015representation}, i.e., the weight polytope is full-dimensional. More concretely, for each weighted game $v$ and each $\varepsilon\in\mathbb{R}_{>0}$ 
there are $w,w'\in \mathsf{W}(v)$ and $q,q'\in(0,1]$ such that $v=[q;w]=[q';w']$ and 
$$
  \operatorname{diam}(\mathsf{W}(v))-\varepsilon\le \Vert w-w'\Vert_1  \le \operatorname{diam}(\mathsf{W}(v)).
$$
Given the indicated linear programming formulation, $\operatorname{diam}(\mathsf{W}(v))$ can be computed in polynomial time (in terms of the 
number of minimal winning and maximal losing coalitions). The same is true if we replace $\Vert\cdot\Vert_1$ by the maximum norm 
$\Vert x\Vert_\infty=\max\{x_i \mid 1\le i\le n\}$ for $x\in\mathbb{R}^n$. We denote the corresponding diameter by  
$\operatorname{diam}^\infty(\mathsf{W}(v))$. For an arbitrary $p$-norm $\Vert x\Vert_p:=\left(\sum_i x_i^p\right)^{1/p}$ with $1<p<\infty$, we can obtain lower and upper 
bounds via $\Vert x\Vert_\infty\le\Vert x\Vert_p\le\Vert x\Vert_1$, so that we restrict ourselves to the corresponding two distance functions. The bound 
$\Vert x\Vert_{\infty}\le\Vert x\Vert_1$ can be slightly improved in our context.

\begin{Lemma}
  \label{lemma_improved_relation_infty_1}
  For $w,w'\in\mathbb{R}^n_{\ge 0}$ with $\Vert w\Vert_1=\Vert w'\Vert_1=1$, we have 
  $\Vert w-w'\Vert_{\infty}\le \frac{1}{2}\Vert w-w'\Vert_1$.
\end{Lemma}
\begin{proof}
  With $S:=\{1\le i\le n\mid w_i\le w'_i\}$ and $A:=\sum_{i\in S} \left(w'_i-w_i\right)$, 
  $B:=\sum_{i\in N\backslash S} \left(w_i-w'_i\right)$, where $N=\{1,\dots,n\}$, we have $A-B=0$ since $\Vert w\Vert_1=\Vert w'\Vert_1$ 
  and $w,w'\in\mathbb{R}_{\ge 0}^n$. Thus, $\Vert w-w'\Vert_1=2A$ and $\Vert w-w'\Vert_{\infty}\le \max\{A,B\}=A$.\hfill{$\square$}
\end{proof}

What can be said about $\operatorname{diam}(\mathsf{W}(v))$ and $\operatorname{diam}^\infty(\mathsf{W}(v))$ in general without solving 
the specific linear programs? Obviously, we have $\operatorname{diam}(\mathsf{W}(v))\le 2$ and $\operatorname{diam}^\infty(\mathsf{W}(v))\le 1$. 
These bounds are asymptotically attained for $n\ge 2$ and $v=[n;(1,\dots,1)]$, i.e., for any $0<\varepsilon<\frac{1}{n}$ we can set 
$w=(1-(n-1)\cdot\varepsilon,\dots,\varepsilon)$, $w'=(\varepsilon,\dots,\varepsilon,1-(n-1)\cdot \varepsilon)$, $q=q'=1-\varepsilon$ so that 
$v=[q;w]=[q';w']$, $\Vert w-w'\Vert_1=2\cdot(1-n\varepsilon)$, and $\Vert w-w'\Vert_\infty=1-2\varepsilon$. In other words, 
$(1,0,\dots,0),(0,\dots,0,1)\in \mathsf{W}([n;1,\dots,1])$ attain the desired distances. For the weighted game $v$ with 
$n=1$ players we have  $\operatorname{diam}(\mathsf{W}(v))=\operatorname{diam}^\infty(\mathsf{W}(v))=0$ since $\mathsf{W}(v)=\{(1)\}$.

In order to obtain tighter bounds for the diameter of the weight polytope we need more information besides the number of players.
Given an exemplary normalized representation $(q;w)$, we study \textit{key parameters} like the relative quota $q\in(0,1]$, i.e., the quota 
of a normalized representation, or the maximum 
relative weight $\Delta(w):=\Vert w\Vert_\infty\in(0,1]$, where we write $\Delta$ whenever $w$ is clear from the context.  
Besides this, also more sophisticated invariants of weight vectors have been  studied in applications. The so-called \textit{Laakso-Taagepera index} 
a.k.a.\ \textit{Herfindahl-Hirschman index}, c.f.\ \cite{laakso1981proportional}, is used in Industrial Organization to measure the concentration of 
firms in a market, see, e.g., \cite{curry1983industrial}, and given by
$$
    L(w)
    =\left(\sum\limits_{i=1}^{n}w_i\right)^2 / \sum\limits_{i=1}^{n} w_i^2.
$$
for $w\in\mathbb{R}_{\ge 0}^n$ with $w\neq 0$. In general we have $1\le L(w)\le n$. If the weight vector $w$ is normalized, then the formula simplifies to 
$L(w)=1/\sum_{i=1}^n w_i^2$. Under the name {\lq\lq}effective number of parties{\rq\rq} the index is widely used in political science 
to measure party fragmentation, see, e.g., \cite{laakso1979effective}.  
However, we observe the following relations between the maximum relative weight 
$\Delta=\Delta(w)$ and the Laakso-Taagepera index $L(w)$:
\begin{Lemma}
  \label{lemma_relation_maximum_laakso_taagepera}
  For $w\in\mathbb{R}_{\ge 0}^n$ with $\Vert w\Vert_1=1$, we have
  $$
    \frac{1}{\Delta}\le
    \frac{1}{\Delta\left(1-\alpha(1-\alpha)\Delta\right)}\le L(w)\le
    \frac{1}{\Delta^2+\frac{(1-\Delta)^2}{n-1}}\le\frac{1}{\Delta^2}
  $$
  for $n\ge 2$, where $\alpha:=\frac{1}{\Delta}-\left\lfloor\frac{1}{\Delta}\right\rfloor\in[0,1)$. 
  If $n=1$, then  $\Delta=L(w)=1$. 
\end{Lemma}
\begin{proof}
  Optimize $\sum\limits_{i=1}^n w_i^2$ with respect to the constraints $w\in\mathbb{R}^n$, $\Vert w\Vert_1=1$, and $\Delta(w)=\Delta$, 
  see the appendix 
  for the technical details.
\end{proof} 

So, any lower or upper bound involving $L(w)$ can be replaced by a bound involving $\Delta$ instead. Since $\Delta$ has nicer analytical 
properties and requires less information on $w$, we stick to $\Delta$ in the following. We remark that there are similar inequalities for 
other indices measuring market concentration. Upper bounds on $\operatorname{diam}(\mathsf{W}(v))$, in terms of $n$, $q$, and $\Delta$, 
will be given in Section~\ref{sec_upper_bounds} and worst case lower bounds for $\operatorname{diam}(\mathsf{W}(v))$ and 
$\operatorname{diam}^\infty(\mathsf{W}(v))$ will be given in Section~\ref{sec_lower_bounds}.

\section{Worst case lower bounds for the diameter of the weight polytope}
\label{sec_lower_bounds}

For integers $1\le k\le s$ and $t\ge 0$ we denote by $v_{k,s,t}$ the weighted game with $s$ players of weight one, $t$ players of weight 
zero, and a quota of $k$, i.e., $v_{k,s,t}=[k;1,\dots,1,0,\dots 0]$. Players $1,\dots,s$ are pairwise equivalent as well as players $s+1,\dots s+t$, 
which are null players. If $k=1$, then each player $1\le i\le s$ is a passer. First we study lower bounds for the diameter of those weighted games.

\begin{Lemma}
  \label{lemma_diameter}
  For integers $1\le k< s$ and $t\ge 0$ we have 
  $$
    \operatorname{diam}(\mathsf{W}(v_{k,s,t}))\ge \max\left\{\tfrac{1}{10k},\tfrac{1}{10(s-k)}\right\}
    \quad\text{and}\quad 
    \operatorname{diam}^\infty(\mathsf{W}(v_{k,s,t}))\ge \tfrac{1}{s}.
  $$  
\end{Lemma}
\begin{proof}
   Let $S=\{1,\dots s\}$ and $T=\{s+1,\dots,s+t\}$. We start with the lower bound for $\operatorname{diam}(\mathsf{W}(v_{k,s,t}))$. 
   If $s$ is even, then we set $S_1=\{1,\dots,s/2\}$, $S_0=\emptyset$, and $S_{-1}=\{s/2+1,\dots,s\}$.  
  If $s$ is odd, then we set $S_{1}=\{1,\dots,(s-1)/2\}$, $S_0=\{(s+1)/2\}$, and $S_{-1}=\{(s+3)/2,\dots,s\}$. Let $0\le \gamma\le\frac{1}{s}$ 
  be a parameter that we specify latter depending on further case differentiations. With this, we set $w_i=\frac{1}{s}+\gamma$ for all $i\in S_1$, 
  $w_i=\frac{1}{s}$ for all $i\in S_0$, $w_i=\frac{1}{s}-\gamma$ for all $i\in S_{-1}$, $w_i=\bar{w}_i=0$ for all $i\in T$, and $\bar{w_i}=w_{s+1-i}$ for all $i\in S$. It is easily 
  verified that $w\in\mathbb{R}_{\ge 0}^{s+t}$ and $\Vert w\Vert_1=1$. In order to conclude $w\in \mathsf{W}(v_{k,s,t})$ it suffices to 
  check $w(U)+w(T)=w(U)\le w(V)$ for all $U,V\subseteq S$ with $|U|=k-1$ and $|V|=k$. Since $\bar{w}$ is a permutation of $w$, $w\in \mathsf{W}(v_{k,s,t})$  
  implies $\bar{w}\in \mathsf{W}(v_{k,s,t})$, so that 
  $$
    \operatorname{diam}(\mathsf{W}(v_{k,s,t}))\ge \Vert w-\bar{w}\Vert_1=2\gamma\cdot |S_1|=2\gamma\cdot\left\lfloor\frac{s}{2}\right\rfloor 
    \ge \frac{\gamma s}{2},
  $$ 
  where we have used $s\ge 2$ for the last inequality.
   
  If $k\le \frac{s+1}{2}$ we set $\gamma=\frac{1}{s(2k-1)}\le \frac{1}{s}$. For $U,V\subseteq S$ with $|U|=k-1$ and $|V|=k$ 
  we have $w(U)\le (k-1)\cdot \left(\frac{1}{s}+\gamma\right)$ and $w(V)\ge k\cdot\left(\frac{1}{s}-\gamma\right)$ so that $w(U)\le w(V)$ 
  and $\operatorname{diam}(\mathsf{W}(v_{k,s,t}))\ge\tfrac{1}{4k}\ge \tfrac{1}{10(s-k)}$.
  
  If $k\ge \frac{s+2}{2}$ we set $\gamma=\frac{1}{s(2s+3-2k)}\le\frac{1}{s}$. For $U,V\subseteq S$ with $|U|=k-1$ and $|V|=k$ 
  we have 
  $$
    w(U)\le \frac{s}{2}\cdot \left(\frac{1}{s}+\gamma\right) +\frac{1}{s}+ \left(k-1-\frac{s}{2}-1\right)\cdot \left(\frac{1}{s}-\gamma\right)
  $$
  and
  $$
    w(V)\ge \frac{s}{2}\cdot \left(\frac{1}{s}-\gamma\right) +\frac{1}{s}+ \left(k-\frac{s}{2}-1\right)\cdot \left(\frac{1}{s}+\gamma\right)
  $$
  so that $w(U)\le w(V)$ and 
  $$
    \operatorname{diam}(\mathsf{W}(v_{k,s,t}))\ge \frac{\gamma s}{2}\ge \frac{1}{2(2(s-k)+3)}\overset{s-k\ge 1}{\ge} \frac{1}{10(s-k)}\ge \frac{1}{10k}.
  $$
  
  Next we consider the lower bound for $\operatorname{diam}^\infty(\mathsf{W}(v_{k,s,t}))$. We set $\gamma=\frac{1}{2s}$, $w_1=\bar{w}_2=\tfrac{1}{s}+\gamma$, 
  $w_2=\bar{w}_1=\tfrac{1}{s}-\gamma$, $w_i=\bar{w}_i=\frac{1}{s}$ for all $3\le i\le s$, and $w_i=\bar{w}_i=0$ for all $i\in T$. It is easily 
  verified that $w\in\mathbb{R}_{\ge 0}^{s+t}$ and $\Vert w\Vert_1=1$. In order to conclude $w\in \mathsf{W}(v_{k,s,t})$ it suffices to 
  check $w(U)+w(T)=w(U)\le w(V)$ for all $U,V\subseteq S$ with $|U|=k-1$ and $|V|=k$. The latter follows from $w(U)\le \frac{k-1}{s}+\gamma$ and 
  $w(V)\ge \frac{k}{s}-\gamma$. Since $\bar{w}$ is a permutation of $w$, we also have $\bar{w}\in \mathsf{W}(v_{k,s,t})$, so that
  $$
    \operatorname{diam}^\infty(\mathsf{W}(v_{k,s,t}))\ge \Vert w-\bar{w}\Vert_{\infty}=2\gamma=\frac{1}{s}.
  $$
\end{proof}

For the excluded cases $k=s$ we have:
\begin{Lemma}
  \label{lemma_diameter2}
  For integers $s\ge 1$ and $t\ge 0$ with $t+s\ge 2$ we have 
  $$
  \operatorname{diam}(\mathsf{W}(v_{s,s,t}))\ge \frac{2}{3}
    \quad\text{and}\quad 
    \operatorname{diam}^\infty(\mathsf{W}(v_{s,s,t}))\ge \frac{1}{3}.
  $$ 
\end{Lemma}
\begin{proof}
  Let $0<\varepsilon<\tfrac{1}{s}$ be arbitrary. If $s\ge 2$ we choose $w_1=\bar{w}_s=1-(s-1)\varepsilon$, $w_i=\bar{w}_{s+1-i}=\varepsilon$ for all $2\le i\le s$, 
  and $w_i=\bar{w}_i=0$ for all $s+1\le i\le s+t$. We can easily check $w,\bar{w}\in \mathsf{W}(v_{s,s,t})$. Since $\Vert w-\bar{w}\Vert_1=2\cdot (1-s\varepsilon)$
  and $\Vert w-\bar{w}\Vert_{\infty}=1-s\varepsilon$ we have $\operatorname{diam}(\mathsf{W}(v_{s,s,t}))\ge \frac{2}{3}$ and 
  $\operatorname{diam}^\infty(\mathsf{W}(v_{s,s,t}))\ge \frac{1}{3}$ using $\varepsilon<\tfrac{2}{3s}$.
  
  If $s=1$ then we consider $w=(1,0,0,\dots, 0)\in \mathsf{W}(v_{1,1,t})$ and $\bar{w}=(\tfrac{2}{3},\tfrac{1}{3},0,\dots, 0)\in \mathsf{W}(v_{1,1,t})$. Thus,  
  $\operatorname{diam}(\mathsf{W}(v_{1,1,t}))\ge \Vert w-\bar{w}\Vert_1=\tfrac{2}{3}$ and 
  $\operatorname{diam}^\infty(\mathsf{W}(v_{1,1,t}))\ge \Vert w-\bar{w}\Vert_{\infty}=\tfrac{1}{3}$.
\end{proof}

Next we show that for a given relative quota $q\in(0,1]$ or a given maximum relative weight $\Delta\in(0,1]$ we can construct 
a weighted game $v$, for any suitably large number of players, with matching representation such that $\operatorname{diam}(\mathsf{W}(v))$ 
is lower bounded by a positive constant independent of $q$ or $\Delta$. Actually, we construct two representations of the same weighted 
game and give a lower bound for the distance between the two normalized weight vectors.

\begin{Lemma}
  \label{lemma_lb_approximation_q}
  For each $q\in(0,1]$ there exists a weighted game $v=[q;w]=[q;\bar{w}]$ with $n\ge 2$ players, where 
  $w,\bar{w}\in\mathbb{R}^n_{\ge 0}$, and $\Vert w\Vert_1=\Vert \bar{w}\Vert_1=1$, such that   
  $\Vert w-\bar{w}\Vert_{\infty}\ge \frac{1}{3}$ and $\Vert w-\bar{w}\Vert_{1}\ge \frac{2}{3}$. 
\end{Lemma} 
\begin{proof}
  We give general constructions for different ranges of $q$:
  \begin{itemize}
    \item $\frac{2}{3}<q\le 1$: $w=\left(\frac{2}{3},\frac{1}{3},0,\dots,0\right)$, $\bar{w}=\left(\frac{1}{3},\frac{2}{3},0,\dots,0\right)$; 
    \item $\frac{1}{3}<q\le \frac{2}{3}$: $w=\left(\frac{2}{3},\frac{1}{3},0,\dots,0\right)$, $\bar{w}=\left(1,0,\dots,0\right)$;
    \item $0<q\le \frac{1}{3}$: $w=\left(\frac{2}{3},\frac{1}{3},0,\dots,0\right)$, $\bar{w}=\left(\frac{1}{3},\frac{2}{3},0,\dots,0\right)$.
  \end{itemize}
\end{proof}

\begin{Lemma}
  \label{lemma_lb_approximation_delta}
  Let $\Delta\in(0,1]$ and $n\ge \frac{1}{\Delta}+1$. There exist $w,\bar{w}\in\mathbb{R}_{\ge 0}^n$, $q,\bar{q}\in(0,1]$ with $\Vert w\Vert_1=\Vert \bar{w}\Vert_1=1$, 
  $\Delta(w)=\Delta$, $[q;w]=[\bar{q},\bar{w}]$, and $\tfrac{1}{2}\cdot\Vert w-\bar{w}\Vert_1\ge \Vert w-\bar{w}\Vert_\infty \ge \tfrac{1}{7}$.
\end{Lemma}
\begin{proof}
  We set $s=\left\lfloor\tfrac{1}{\Delta}\right\rfloor\ge 1$ and $t=n-s\ge 1$, since $n\ge \frac{1}{\Delta}+1\ge s+1$. For 
  $w=(\Delta,\dots,\Delta,1-s\Delta,0,\dots,0)\in\mathbb{R}_{\ge 0}^n$,  with $s$ entries being equal to $\Delta$, we have 
  $\Delta(w)=\Delta$ and $[q;w]=v_{s,s,t}$ for $0<q=s\Delta\le 1$. Due to Lemma~\ref{lemma_diameter2} we have 
  $\operatorname{diam}^\infty(\mathsf{W}(v_{s,s,t}))\ge \frac{1}{3}$,  so that the triangle inequality implies the existence of a vector $w'\in \mathsf{W}(v_{s,s,t})$ 
  with $\Vert w-w'\Vert_\infty\ge \frac{1}{6}$. If $w'$ is on the boundary  of $\mathsf{W}(v_{s,s,t})$ we slightly perturb $w'$ to $\bar{w}$ in the interior of 
  $\mathsf{W}(v_{s,s,t})$ and complete it to a representation $(\bar{q},\bar{w})$ with $\bar{q}\in(0,1]$, $[q;w]=[\bar{q},\bar{w}]$, and 
  $\Vert w-\bar{w}\Vert_\infty\ge \frac{1}{7}$. The inequality $\tfrac{1}{2}\cdot \Vert w-\bar{w}\Vert_1\ge \vert w-\bar{w}\Vert_\infty$ follows from 
  Lemma~\ref{lemma_improved_relation_infty_1}.
\end{proof}

By a tailored construction we can obtain a slightly more general result:
\begin{Lemma}
  \label{lemma_lb_approximation_delta_1}
  For each $\Delta\in(0,1)$ there exists a weighted game $v=[q;w]=[q;\bar{w}]$ with $n\ge \frac{4}{3\Delta}+6$ players, where 
  $q\in(0,1)$, $w,\bar{w}\in\mathbb{R}^n_{\ge 0}$, $\Delta(w)=\Delta(\bar{w})=\Delta$, and  
  $\Vert w\Vert_1=\Vert \bar{w}\Vert_1=1$, such that $\Vert w-\bar{w}\Vert_{1}\ge \frac{2}{3}$ and 
  $\Vert w-\bar{w}\Vert_{\infty}\ge\Delta/2$.
\end{Lemma}
\begin{proof} 
  If $\Delta\ge \frac{2}{3}$, we can consider a weighted game with two passers and $n-2$ null players. One representation is given by 
  $q=1-\Delta$ and $w=(\Delta,1-\Delta,0,\dots,0)$. Of course we can swap the weights of the first two players and obtain a 
  second representation given by quota $q$ an weight vector $\bar{w}=(1-\Delta,\Delta,0,\dots,0)$. With this, we compute 
  $\Vert w-\bar{w}\Vert_{1}=2\cdot (2\Delta-1)\ge \frac{2}{3}$ and
  $\Vert w-\bar{w}\Vert_{\infty}=2\Delta-1\ge \Delta/2$. 

  If $0<\Delta<\frac{2}{3}$, we define an integer $a:=\left\lfloor\frac{2}{3\Delta}\right\rfloor\ge 1$ and consider a weighted game 
  with $2a$ passers and $n-2a$ null players. One representation is given by $q=\Delta/2$, $w_{2i-1}=\Delta$, $w_{2i}=\Delta/2$ for $1\le i\le a$, 
  $w_{2a+1}=w_{2a+3}=w_{2a+5}=\frac{1}{3}-\frac{a\Delta}{2}\ge 0$, $w_{2a+2}=w_{2a+4}=w_{2a+6}=0$, and $w_i=0$ for all $2a+7\le i\le n$.
  By assumption we have $n\ge \tfrac{4}{3\Delta}+6\ge 2a+6$ and the first $2a$ players are obviously passers. By checking $0\le \frac{1}{3}-\frac{a\Delta}{2}<\frac{\Delta}{2}$ 
  we conclude that the remaining players are null players and have a non-negative weight. By construction, the weights of the $n$ players 
  sum up to one. Changing the weights of player $2i-1$ and player $2i$ for $1\le i\le a$ does not change the game so that we obtain 
  a second representation with quota $q$ and weights $\bar{w}_{2i}=\Delta$, $\bar{w}_{2i-1}=\Delta/2$ for $1\le i\le a$, 
  $\bar{w}_{2a+2}=\bar{w}_{2a+4}=\bar{w}_{2a+6}=\frac{1}{3}-\frac{a\Delta}{2}\ge 0$, 
  $w_{2a+1}=w_{2a+3}=w_{2a+4}=\bar{w}_{2a+1}=\bar{w}_{2a+2}=\bar{w}_{2a+3}=0$, and $\bar{w}_i=0$ for all $2a+7\le i\le n$. 
  With this, we have $\Vert w-\bar{w}\Vert_{1}=a\Delta +2-3a\Delta=2(1-a\Delta)\ge\frac{2}{3}$ and $\Vert w-\bar{w}\Vert_{\infty}=\Delta/2$.
\end{proof}   
For each  $w,\bar{w}\in\mathbb{R}_{\ge 0}^n$ with $\Delta(w)=\Delta(\bar{w})$, we obviously have $\Vert w-\bar{w}\Vert_\infty\le \Delta(w)$. 
So, a constant lower bound for the $\Vert\cdot\Vert_\infty$-distance can only exist if we slightly weaken the assumptions as done in Lemma~\ref{lemma_lb_approximation_delta}. 

In some applications only weighted games with a quota of at least one half are considered, which clashes with some of our constructions 
in the proofs of the previous lemmas. However, by considering the dual of a given weighted game we can turn a quota below one half to 
a quota above one half, see Lemma~\ref{lemma_dual_weights}. So, instead of small quotas we get large quotas.

So, either knowing the relative quota or the maximum relative weight is not sufficient in order to deduce a non-constant 
upper bound on the diameter of the weight polytope for a suitably large number of players. However, as we will see in the next 
section, knowing the relative quota and the maximum relative weight is indeed sufficient for such an upper bound, see 
Theorem~\ref{thm_weighted_representation}.  Our next aim is to show that this upper bound is tight up to a constant.

\begin{Lemma}
  \label{lemma_lb_diam_representation_polytop}
  For each $0<q<1$, $0<\Delta\le 1$, and each integer $n\ge \tfrac{1}{\Delta}+2$ there exist weight vectors 
  $w,\bar{w}\in\mathbb{R}^n_{\ge 0}$ with $\Vert w\Vert_1=\Vert\bar{w}\Vert_1=1$, $\Delta(w)=\Delta$  and a quota $0<\bar{q}\le 1$ 
  with $[q;w]=[\bar{q};\bar{w}]$ such that 
  $$
    \Vert w-\bar{w} \Vert_1\ge \frac{1}{200}\cdot \min\left\{2, \frac{4\Delta}{\min\{q,1-q\}}\right\}.
  $$ 
  Under the same assumptions there exist weight vectors $w,\bar{w}\in\mathbb{R}^n_{\ge 0}$ with $\Vert w\Vert_1=\Vert\bar{w}\Vert_1=1$, 
  $\Delta(w)=\Delta$  and a quota $0<\bar{q}\le 1$ with $[q;w]=[\bar{q};\bar{w}]$ such that 
  $\Vert w-\bar{w} \Vert_\infty\ge \frac{\Delta}{5}$. 
\end{Lemma}
\begin{proof}
  We set $a=\left\lfloor\tfrac{1}{\Delta}\right\rfloor\ge 1$ and choose the unique integer $b$ with $b\Delta<q$ and $(b+1)\Delta\ge q$. 
  With this we set $k=b+1\ge 1$ and $w=(\Delta,\dots,\Delta,1-a\Delta,0,\dots,0)$, where $0\le 1-a\Delta<\Delta$, so that $w\in\mathbb{R}_{\ge 0}^n$ 
  and $\Vert w\Vert_1=1$. If $b\Delta+(1-a\Delta)<q$ we set $s=a$ and $s=a+1$ otherwise, so that $[q;w]=v_{k,s,n-s}$. Note that $n-s\ge 1$.   

  If $k=s$, then Lemma~\ref{lemma_diameter2} gives $\operatorname{diam}(\mathsf{W}(v_{s,s,t}))\ge \frac{2}{3}$, so that the triangle inequality implies 
  the existence of a vector $w'\in \mathsf{W}(v_{s,s,t})$  with $\Vert w-w'\Vert_1\ge \frac{1}{3}$. If $k<s$, then Lemma~\ref{lemma_diameter} gives 
  $\operatorname{diam}(\mathsf{W}(v_{k,s,t}))\ge \max\left\{\frac{1}{10k},\frac{1}{10(s-k)}\right\}$, so that the triangle inequality implies 
  the existence of a vector $w'\in \mathsf{W}(v_{k,s,t})$  with 
  $$
    \Vert w-w'\Vert_1\ge \max\left\{\frac{1}{20k},\frac{1}{20(s-k)}\right\}=\frac{1}{20s}\cdot\frac{1}{\min\left\{\frac{k}{s},\frac{s-k}{s}\right\}}.
  $$  
  In the following we make several case distinctions for the subcase $k<s$.
  
  If $k=1$ or $s-k=1$, then $\Vert w-w'\Vert_1\ge\frac{1}{20}$. In the following we assume $k\ge 2$ and $s-k\ge 2$. 
  By construction we have $\tfrac{k}{2}\le (k-1)\Delta<q$, $k\Delta\ge q$, and $(s-1)\Delta\le 1$, so that $k<\frac{2q}{\Delta}$, 
  $\tfrac{s-k}{2}\Delta\le (s-1)\Delta-k\Delta\le 1-q$ and $s-k\le \frac{2(1-q)}{\Delta}$. 
  
  If $k\le s-k$, i.e., $2k\le s$, then $q\le \tfrac{1}{2}$ and
  $$
    \Vert w-w'\Vert_1\ge \frac{1}{20s}\cdot\frac{1}{\min\left\{\frac{k}{s},\frac{s-k}{s}\right\}}
    =\frac{1}{20k}\ge \frac{1}{40}\cdot \frac{\Delta}{q}=\frac{1}{40}\cdot \frac{\Delta}{\min\{q,1-q\}}.
  $$ 
  If $k> s-k$, i.e., $2k>s$, then $q> \tfrac{1}{2}$ and
  $$
    \Vert w-w'\Vert_1\ge \frac{1}{20s}\cdot\frac{1}{\min\left\{\frac{k}{s},\frac{s-k}{s}\right\}}
    =\frac{1}{20(s-k)}\ge \frac{1}{40}\cdot \frac{\Delta}{1-q}=\frac{1}{40}\cdot \frac{\Delta}{\min\{q,1-q\}}. 
  $$ 
  Thus,
  $$
    \Vert w-w' \Vert_1\ge \frac{1}{160}\cdot \min\left\{2, \frac{4\Delta}{\min\{q,1-q\}}\right\}
  $$
  in all cases. If $w'$ is on the boundary of $\mathsf{W}(v_{k,s,n-s})$, then we slightly perturb $w'$ to $\bar{w}$ in the interior of 
  $\mathsf{W}(v_{k,s,n-s})$ and choose a quota $\bar{q}\in(0,1]$ such that $[\bar{q};\bar{w}]=v_{k,s,n-s}$. This gives the statement for the $\Vert\cdot\Vert_1$-distance, 
  if the pertubation is small enough to be covered by our decrease of the factor $\tfrac{1}{160}$ to $\tfrac{1}{200}$.
  
  For the $\Vert\cdot\Vert_\infty$-distance we choose $w$ with $[q;w]=v_{k,s,n-s}$ as above. If $k=s$, then Lemma~\ref{lemma_diameter2} gives 
  $\operatorname{diam}^\infty(\mathsf{W}(v_{s,s,t}))\ge \frac{1}{3}$, so that the triangle inequality implies 
  the existence of a vector $w'\in \mathsf{W}(v_{s,s,t})$  with $\Vert w-w'\Vert_\infty\ge \frac{1}{6}$. If $k<s$, then Lemma~\ref{lemma_diameter} gives 
  $\operatorname{diam}^\infty(\mathsf{W}(v_{k,s,t}))\ge \frac{1}{s}$, so that the triangle inequality implies 
  the existence of a vector $w'\in \mathsf{W}(v_{k,s,t})$  with $\Vert w-w'\Vert_\infty\ge \tfrac{1}{2s}$.  For $s=1$ this gives $\Vert w-w'\Vert_\infty\ge \tfrac{1}{2}$. 
  For $s\ge 2$ we have $s\le\tfrac{2}{\Delta}$ so that $\Vert w-w'\Vert_\infty\ge \tfrac{\Delta}{4}$. Since $\Delta\le 1$ we have $\Vert w-w'\Vert_\infty\ge \tfrac{\Delta}{4}$ 
  in all cases, so that the stated result follows possibly by a perturbation.
\end{proof}

\section{Upper bounds for the diameter of the weight polytope}
\label{sec_upper_bounds}

Before we start to upper bound $\operatorname{diam}(\mathsf{W}(v))$ in terms of $\Delta$ and $q$, we provide a slightly more general 
result.

\begin{Lemma}
  \label{lemma_qdelta_bound_winning}
  Let $w\in\mathbb{R}^n_{\ge 0}$ with $\Vert w\Vert_1=1$ for an integer $n\in\mathbb{N}_{>0}$ and $0<q<1$. For each $x\in\mathbb{R}^n_{\ge 0}$ with 
  $\Vert x\Vert_1=1$ and $x(S)=\sum_{s\in S}x_s\ge q$ for every winning coalition $S$ of $[q;w]$, we have
  $$
    \Vert w-x\Vert_1 \le \frac{2\Delta}{\min\{q+\Delta,1-q\}}\le\frac{2\Delta}{\min\{q,1-q\}},
  $$ 
  where $\Delta=\Delta(w)$.
\end{Lemma}
\begin{proof}
  Consider a winning coalition $T$ such that $x(T)$ is minimal and invoke $x(T)\ge q$, see the appendix 
  for the technical details.
\end{proof}

From Lemma~\ref{lemma_qdelta_bound_winning} we can directly conclude:  
\begin{Corollary}
  Let $w,\bar{w}\in\mathbb{R}^n_{\ge 0}$ with $\Vert w\Vert_1=\Vert \bar{w}\Vert_1=1$ for an integer $n\in\mathbb{N}_{>0}$ and 
  $0<q, \bar{q}<1$. If $[q;w]=[\bar{q};\bar{w}]$, then we have 
  $$
    \Vert w-\bar{w}\Vert_1 \le \max\left\{ \frac{2\Delta(w)}{\min\{q,1-q\}},
    \frac{2\Delta(\bar{w})}{\min\{\bar{q},1-\bar{q}\}}\right\}
    \le\frac{2\Delta(w)}{\min\{q,1-q\}}+\frac{2\Delta(\bar{w})}{\min\{\bar{q},1-\bar{q}\}}.
  $$
\end{Corollary}

Unfortunately, this does not allow us to derive an upper bound of $\Vert w-\bar{w}\Vert_1$ which only depends on $q$ and $\Delta(w)$. 
However, we can obtain the following analog of Lemma~\ref{lemma_qdelta_bound_winning} for losing instead of winning coalitions.

\begin{Lemma}
  \label{lemma_qdelta_bound_losing}
  Let $w\in\mathbb{R}^n_{\ge 0}$ with $\Vert w\Vert_1=1$, $\Delta=\Delta(w)$, and $0<q<1$. For each $x\in\mathbb{R}^n_{\ge 0}$ with 
  $\Vert x\Vert_1=1$ and $x(S)=\sum_{s\in S}x_s\le q$ for every losing coalition $S$ of $[q;w]$, we have 
  $$
    \Vert w-x\Vert_1 \le\frac{4\Delta}{\min\{q,1-q\}}.
  $$
  Moreover, if $q>\Delta$, then  
  $
    \Vert w-x\Vert_1 \le \frac{2\Delta}{\min\{q-\Delta,1-q+\Delta\}}\le\frac{2\Delta}{\min\{q-\Delta,1-q\}}
  $.
\end{Lemma}
\begin{proof}
  Consider a losing coalition $T$ such that $x(T)$ is maximal and invoke $x(T)\le q$. Technical details are provided in the appendix.
\end{proof}

\begin{Theorem}
  \label{thm_weighted_representation}
  Let $w,\bar{w}\in\mathbb{R}^n_{\ge 0}$ with $\Vert w\Vert_1=\Vert \bar{w}\Vert_1=1$, $\Delta=\Delta(w)$, and 
  $0<q, \bar{q}<1$. If $[q;w]=[\bar{q};\bar{w}]$, then we have 
  $$
    \Vert w-\bar{w}\Vert_1\le \min\left\{2,\frac{4\Delta}{\min\{q,1-q\}}\right\}\le\frac{4\Delta}{\min\{q,1-q\}},
  $$
  i.e., $\operatorname{diam}(\mathsf{W}([q;w]))\le \frac{4\Delta(w)}{\min\{q,1-q\}}$. 
  Moreover, if $q>\Delta$, then we have 
  $$
     \Vert w-\bar{w}\Vert_1\le \frac{2\Delta}{\min\{q-\Delta,1-q\}}.
  $$
\end{Theorem}
\begin{proof}
   In Section~\ref{sec_weight_polytope} we have observed $\Vert w-\bar{w}\Vert_1\le 2$. 
   If $\bar{q}\ge q$, then $\bar{w}(S)\ge \bar{q}\ge q$ for every winning coalition $S$ of $[q;w]$. Here, we can apply 
   Lemma~\ref{lemma_qdelta_bound_winning}. Otherwise we have $\bar{w}(T)<\bar{q}<q$ for every losing coalition $T$ of $[q;w]$ 
   and Lemma~\ref{lemma_qdelta_bound_losing} applies.
\end{proof}

As an example we consider the normalized weight vector $w=\tfrac{1}{120}\cdot(15,14,\dots,1)$ and the quota $\tfrac{3}{5}$. Let $(\bar{q};\bar{w})$ 
be another normalized representation of the weighted game $[q;w]$, then the first bound gives $\Vert w-\bar{w}\Vert_1\le \tfrac{5}{4}$. Since 
$\Delta=\tfrac{1}{8}>q$, also the second bound applies yielding $\Vert w-\bar{w}\Vert_1\le \tfrac{5}{8}$. We remark that for this specific example 
the diameter $\operatorname{diam}(\mathsf{W}([q;w]))$ is much smaller than $\tfrac{5}{8}$.

\section{Applications}
\label{sec_applications}

A power index $\varphi$ is a mapping from the set of weighted games on $n$ players into $\mathbb{R}_{\ge 0}^n$. We call $\varphi$ 
efficient if $\Vert\varphi(v)\Vert_1=1$ for all weighted games $v$. The 
difference $\Vert w-\varphi([q;w])\Vert_1$ between relative weights and the corresponding power distribution is 
studied in the literature, see e.g.\ \cite{dubey1979mathematical,kurz2014nucleolus,neyman1982renewal}. Lemma~\ref{lemma_qdelta_bound_winning} 
is a generalization of \cite[Lemma 1]{kurz2014nucleolus}: if $\varphi$ is the nucleolus, see e.g.\ \cite{schmeidler1969nucleolus}, 
and $0<q<1$ then 
\begin{equation}
  \label{ie_ub_nucleolus}
  \Vert w-\varphi([q;w])\Vert_1 \le\frac{2\Delta(w)}{\min\{q,1-q\}} 
\end{equation}
for all $w\in\mathbb{R}_{\ge 0}^n$ with $\Vert w\Vert_1=1$. From Theorem~\ref{thm_weighted_representation} we directly conclude:
\begin{Corollary}
  \label{cor_ub_general}
  Let $w\in\mathbb{R}^n_{\ge 0}$ with $\Vert w\Vert_1=1$ and $0<q<1$. 
  If an efficient power index $\varphi$ permits the existence of a quota $q'\in(0,1)$ such that $[q';\varphi([q;w])]=[q;w]$, i.e., 
  the power vector of the given weighted game can be completed to a representation of the same game, then
  $$
    \Vert w-\varphi([q;w])\Vert_1
    \le \frac{4\Delta(w)}{\min\{q,1-q\}}.
  $$
\end{Corollary}
\textit{Representation compatibility} of $\varphi$ for $[q;w]$ is automatically satisfied
for the modified nucleolus (modiclus) \cite{sudholter1996modified}, minimum sum representation index \cite{freixas2014minimum} or one of the power indices based 
on averaged representations \cite{kaniovski2015average} for all weighted games and for the Penrose-Banzhaf index for all spherically separable simple games \cite{houy2014geometry}.
The theorem also applies to the bargaining model for weighted games analyzed in \cite{market_value_model}, cf.~\cite{prop_payoffs}. It is unknown 
whether there exists a constant $c\in\mathbb{R}_{>0}$ such that 
\begin{equation}
  \label{ie_ub_ssi}
  \Vert w-\operatorname{SSI}([q;w])\Vert_1
    \le \frac{c\Delta(w)}{\min\{q,1-q\}}.
\end{equation}
holds for the Shapley-Shubik index $\operatorname{SSI}$ and all $w\in\mathbb{R}^n_{\ge 0}$ with $\Vert w\Vert_1=1$ and $0<q<1$. For the Penrose-Banzhaf index 
such a constant $c$ can not exist, see \cite[Proposition 2]{kurz2018note}.

For the other direction we have:
\begin{Lemma}
  \label{lemma_general_lower_distance_bound}
  Let $n\in\mathbb{N}_{>0}$, $q,\bar{q}\in(0,1]$, $w,\bar{w}\in\mathbb{R}^n_{\ge 0}$ with $\left\Vert w\right\Vert_1=
  \left\Vert \bar{w}\right\Vert_1=1$ and $[q;w]=[\bar{q};\bar{w}]$, $\Vert\cdot\Vert$ be an arbitrary norm on 
  $\mathbb{R}^n$ and $\varphi$ be a mapping from the set of weighted games (on $n$ players) into $\mathbb{R}_{\ge 0}^n$,  
  then we have 
  $$
    \max\left\{\left\Vert w-\varphi\left(\left[q;w\right]\right)\right\Vert,
    \left\Vert \bar{w}-\varphi\left(\left[\bar{q};\bar{w}\right]\right)\right\Vert
    \right\}\ge \frac{\left\Vert w-\bar{w}\right\Vert}{2}.
  $$
\end{Lemma} 
\begin{proof}
    Using the triangle inequality yields $\left\Vert w-\varphi\left(\left[q;w\right]\right)\right\Vert+
    \left\Vert \bar{w}-\varphi\left(\left[\bar{q};\bar{w}\right]\right)\right\Vert\ge 
    \left\Vert w-\bar{w}\right\Vert$ from which we can conclude the stated 
    inequality.
\end{proof}

\begin{Proposition}
  \label{prop_impossible}
  Let $\varphi$ be a mapping from the set of weighted games (on $n$ players) into $\mathbb{R}_{\ge 0}^n$.
  \begin{itemize}
    \item[(i)] For each $q\in(0,1]$ and each integer $n\ge 2$ there exists a weighted game $[q;w]$, where $w\in\mathbb{R}^n_{\ge 0}$ and $\Vert w\Vert_1=1$, 
               such that $\Vert w-\varphi([q;w])\Vert_{1}\ge \frac{1}{3}$ and $\Vert w-\varphi([q;w])\Vert_{\infty}\ge \frac{1}{6}$.
    \item[(ii)] For each $\Delta\in(0,1)$ and each integer $n\ge \frac{4}{3\Delta}+6$ there exists a weighted game $[q;w]$, where $q\in(0,1]$, 
                $w\in\mathbb{R}^n_{\ge 0}$, $\Vert w\Vert_1=1$, and $\Delta(w)=\Delta$, 
               such that $\Vert w-\varphi([q;w])\Vert_{1}\ge \frac{1}{3}$, and $\Vert w-\varphi([q;w])\Vert_{\infty}\ge\Delta/4$.
  \end{itemize} 
\end{Proposition}
\begin{proof}
  Combine Lemma~\ref{lemma_general_lower_distance_bound} with lemmas \ref{lemma_lb_approximation_q} and \ref{lemma_lb_approximation_delta_1}.
\end{proof}

\begin{Proposition}
  \label{prop_impossible2}
  Let $\varphi$ be a mapping from the set of weighted games (on $n$ players) into $\mathbb{R}_{\ge 0}^n$. For each $q\in(0,1)$, $\Delta\in(0,1]$, there 
  exist $w,\bar{w}\in\mathbb{R}_{\ge 0}^n$, $\bar{q}\in(0,1]$ with $\Vert w\Vert_1=\Vert\bar{w}\Vert_1=1$, $\Delta(w)=\Delta$, $[q;w]=[\bar{q};\bar{w}]$, 
  and
  $$
    \Vert \bar{w}-\varphi([\bar{q};\bar{w}])\Vert_1\ge \frac{1}{200}\cdot \min\left\{2, \frac{4\Delta}{\min\{q,1-q\}}\right\}.
  $$
\end{Proposition}
\begin{proof}
  We construct $w$ as in the proof of Lemma~\ref{lemma_lb_diam_representation_polytop} and choose integers $k$, $s$, and $t$ such that $[q;w]=v_{k,s,t}$. 
  In the proof of Lemma~\ref{lemma_lb_diam_representation_polytop} we have actually verified 
  $$
    \operatorname{diam}(\mathsf{W}([q;w]))\ge \frac{1}{80}\cdot \min\left\{2, \frac{4\Delta}{\min\{q,1-q\}}\right\}=:\Lambda.
  $$
  Now choose $w',w''\in \mathsf{W}([q;w])$ with $\Vert w'-w''\Vert_1\ge \Lambda$. By the triangle inequality we have either $\Vert w'-\varphi([q;w])\Vert_1\ge \Lambda/2$ 
  or $\Vert w''-\varphi([q;w])\Vert_1\ge \Lambda/2$. By choosing $\bar{w}$ as $w'$ or $w''$ and eventually moving it into the interior of $\mathsf{W}([q;w])$ we obtain 
  the stated result.
\end{proof}

 So, upper bounds for the $\Vert\cdot\Vert_1$-distance between normalized weights and a power distribution, as in Inequality~(\ref{ie_ub_nucleolus} or 
 Inequality~(\ref{ie_ub_ssi}) are tight up to the constant $c$ if only the normalized quota and the normalized maximum weight are taken into account.

\section{Conclusion}
\label{sec_conclusion}

In this paper we have introduced the concept of the diameter of the weight polytope of a weighted game. This number measures how diverse 
two different normalized weight vectors, representing the same given game, can be. In Theorem~\ref{thm_weighted_representation} we have shown that 
$$
    \operatorname{diam}(\mathsf{W}([q;w]))\le \min\left\{2,\frac{4\Delta}{\min\{q,1-q\}}\right\}\le\frac{4\Delta}{\min\{q,1-q\}},
$$
for any $q\in(0,1)$ and any $w\in\mathbb{R}_{\ge 0}^n$ with $\Vert w\Vert_1=1$. Lemma~\ref{lemma_lb_diam_representation_polytop} certifies that this upper 
bound is in general, i.e., in the worst case, tight up to a constant. (This paper traded smaller constants for easier proofs.) The super-exponential growth 
of the number of weighted games (see \cite{zuev1989asymptotics}) indicates that this is not  the case for the majority of weighted games. Thus, it would be 
interesting to determine other parameters of a representation of a weighted game that permit tight upper bounds on the diameter of the corresponding weight 
polytope. Another possible line for future research is to consider games with a priori unions, spatial games, or games with restricted communication. 

As shown in Section~\ref{sec_applications}, there are connections to approximations of power indices by weight vectors. 
Proposition~\ref{prop_impossible2} gives a partial explanation for the conditions of the main theorem of \cite{neyman1982renewal} on a limit 
result for the Shapley-Shubik index. Moreover, for a general power index it shows that upper bounds for the $\Vert\cdot\Vert_1$-distance between normalized 
weights and a power distribution, taking only the normalized quota and the normalized maximum weight into account, as in Corollary~\ref{cor_ub_general}, 
would be tight up to a constant.

\section*{Acknowledgment}
The author would like to thank the anonymous referees of a previous submission for their very helpful remarks and suggestions.

\appendix
\section{Delayed proofs}
\label{sec_evacuated_proofs}
\begin{proof}(Lemma~\ref{lemma_relation_maximum_laakso_taagepera})\\
  For $n=1$, we have $w_1=1$, $\Delta(w)=1$, $\alpha=0$, and $L(w)=1$, so that we assume $n\ge 2$ in the remaining part of the proof.
  For $w_i\ge w_j$ consider $a:=\frac{w_i+w_j}{2}$ and $x:=w_i-a$, so that $w_i=a+x$ and $w_j=a-x$. With this we have 
  $w_i^2+w_j^2=2a^2+2x^2$ and $(w_i+y)^2+(w_j-y)^2=2a^2+2(x+y)^2$. Let us assume that $w^\star$ minimizes $\sum_{i=1}^n w_i^2$ under 
  the conditions $w\in\mathbb{R}_{\ge 0}$, $\Vert w\Vert_1=1$, and 
  $\Delta(w)=\Delta$. (Since the target function is continuous and the feasible set is compact and non-empty, a global minimum indeed 
  exists.) W.l.o.g.\ we assume $w_1^\star=\Delta$. If there are indices $2\le i,j\le n$ with $w_i^\star>w_j^\star$, i.e., $x>0$ in the 
  above parameterization, then we may choose $y=-x$. Setting $w_i':=w_i^\star+y=a=\frac{w_i^\star+w_j^\star}{2}$, $w_j':=w_j^\star-y=a
  =\frac{w_i^\star+w_j^\star}{2}$, and $w_h':=w_h^\star$ for all $1\le h\le n$ with $h\notin\{i,j\}$, we have $w'\in \mathbb{R}_{\ge 0}^n$, 
  $\Vert w'\Vert_1=1$, $\Delta(w')=\Delta$, and $\sum_{h=1}^n \left(w_h'\right)^2=\sum_{h=1}^n \left(w_h^\star\right)^2\,-\,x^2$. 
  Since this contradicts the minimality of $w^\star$, we have $w_i^\star=w_j^\star$ for all $2\le i,j\le n$, so that we conclude 
  $w_i^\star=\frac{1-\Delta}{n-1}$ for all $2\le i\le n$ from $1=\Vert w^\star\Vert_1=\sum\limits_{h=1}^n w_h^\star$.
  Thus, $L(w)\le 1/\left(\Delta^2+\frac{(1-\Delta)^2}{n-1}\right)$, which is tight. Since $\Delta\le 1$ and $n\ge 2$, we have 
  $1/\left(\Delta^2+\frac{(1-\Delta)^2}{n-1}\right)\le \frac{1}{\Delta^2}$, which is tight if and only if $\Delta=1$, i.e., 
  $n-1$ of the weights have to be equal to zero.
  
  Now, let us assume that $w$ maximizes $\sum_{i=1}^n w_i^2$ under the conditions $w\in\mathbb{R}_{\ge 0}$, $\Vert w\Vert_1=1$, and 
  $\Delta(w)=\Delta$. (Due to the same reason a global maximum indeed exists.) Due to $1=\Vert w\Vert_1\le n\Delta$ we have $0<\Delta\le 1/n$, 
  where $\Delta=1/n$ implies $w_i=\Delta$ for all $1\le i\le n$. In that case we have $L(w)=n$ and $\alpha=0$, so that the stated lower bounds 
  for $L(w)$ are valid. In the remaining cases we assume $\Delta>1/n$. If there would exist two indices $1\le i,j\le n$ with $w_i\ge w_j$, 
  $w_i<\Delta$, and $w_j>0$, we may strictly increase the target function by moving weight from $w_j$ to $w_i$ (this corresponds to 
  choosing $y>0$), by an amount small enough to still satisfy the constraints $w_i\le \Delta$ and $w_j\ge 0$. Since $\Delta>0$, we can set 
  $a:=\lfloor 1/\Delta\rfloor\ge 0$ with $a\le n-1$ due to $\Delta>1/n$. Thus, for a maximum solution, we
  have exactly $a$ weights that are equal to $\Delta$, one weight that is equal to $1-a\Delta\ge 0$ (which may indeed 
  be equal to zero), and $n-a-1$ weights that are equal to zero. With this and $a\Delta=1-\alpha\Delta$ we have 
  $
    \sum_{i=1}^n w_i^2=a\Delta^2 (1-a\Delta)^2=\Delta-\alpha\Delta^2+\alpha^2\Delta^2=\Delta(1-\alpha\Delta+\alpha^2\Delta)
    =\Delta\left(1-\alpha(1-\alpha)\Delta\right)\le\Delta
  $.   
  Here, the latter inequality is tight if and only if $\alpha=0$, i.e., $1/\Delta\in\mathbb{N}$.
\end{proof}

\bigskip

\begin{proof}(Lemma~\ref{lemma_qdelta_bound_winning})\\
  We set $N=\{1,\dots,n\}$, $w(U)=\sum_{u\in U} w_u$ and $x(U)=\sum_{u\in U} x_u$ for each $U\subseteq N$. Let $S^+=\{i\in N \mid x_i>w_i\}$ and 
  $S^-=\{i\in N\mid x_i\le w_i\}$, i.e., $S^+$ and $S^-$ partition the set $N$ of players. We have $w(S^+)<1$ since $w(S^+)<x(S^+)\le x(N)=1$, so 
  that $w(S^-)>0$. Define $0\le \delta\le 1$ by $x(S^-)=(1-\delta)w(S^-)$. We have 
  \begin{equation}
    x(S^+)=1-x(S^-)=w(S^+)+w(S^-)-(1-\delta)w(S^-)=w(S^+)+\delta w(S^-)
  \end{equation}   
  and
  \begin{equation}
    \label{eq_dist_delta}
    \Vert w-x\Vert_1
    =\left(x(S^+)-w(S^+)\right)+\left(w(S^-)-x(S^-)\right)=2\delta w(S^-).
  \end{equation}
  Generate a set $T$ by starting at $T=\emptyset$ and successively add a remaining player $i$ in $N\backslash T$ with minimal
  $x_i/w_i$, where all players $j$ with $w_j=0$ are the worst ones. Stop if $w(T)\ge q$. By construction $T$ is a winning coalition 
  of $[q;w]$ with $w(T)<q+\Delta$, since the generating process did not stop earlier and $w_j\le \Delta(w)$ for all $j\in N$.   
  
  If $w(S^-)\ge q$, we have $T\subseteq S^-$ and $x(T)/w(T)\le x(S^-)/w(S^-)=1-\delta$. Multiplying by $w(T)$ and using $w(T)<q+\Delta$ 
  yields
  \begin{equation}
    x(T)\le (1-\delta)w(T)<(1-\delta)(q+\Delta)= (1-\delta)q +(1-\delta)\Delta. 
  \end{equation}
  Since $x(T)\ge q$, as $T$ is a winning coalition, we conclude $\delta<\Delta/(q+\Delta)$. Using this and $w(S^-)<1$ in Equation~(\ref{eq_dist_delta}) 
  yields 
  \begin{equation}
    \Vert w-x\Vert_1<\frac{2\Delta}{q+\Delta}<\frac{2\Delta}{q}.
  \end{equation}  
  
  If $w(S^-)< q$, we have $S^-\subseteq T$, $x(T)=x(S^-)+x(T\backslash S^-)$, $w(T\backslash S^-)>0$, and $w(S^+)>0$. Since 
  $T\backslash S^-\subseteq S^+$, 
  $x(T\backslash S^-)/w(T\backslash S^-)\le x(S^+)/w(S^+)$, so that
  \begin{eqnarray*}
    x(T)&=&x(S^-)+x(T\backslash S^-)\le (1-\delta)w(S^-)+\frac{x(S^+)}{w(S^+)}\cdot\left(w(T)-w(S^-)\right)\\
    &\le& (1-\delta)w(S^-)+\frac{x(S^+)}{w(S^+)}\cdot\left(q+\Delta-w(S^-)\right)\\
    &=& 
    q+\frac{x(S^+)\Delta-(1-q)\delta w(S^-)}{w(S^+)}\\
    &\le & q+\frac{\Delta-(1-q)\delta w(S^-)}{w(S^+)}.
  \end{eqnarray*} 
  Since $x(T)\ge q$, we conclude $(1-q)\delta w(S^-)\le \Delta$, so that 
  $\Vert w-x\Vert_1\le \frac{2\Delta}{1-q}$.     
\end{proof}

\bigskip

\begin{proof}(Lemma~\ref{lemma_qdelta_bound_losing})\\
  If $q\le 2\Delta$, then $\frac{4\Delta}{\min\{q,1-q\}}\ge \frac{4\Delta}{q}\ge 2\ge \Vert x-w\Vert_1$, so that we can assume $q>\Delta$. 

  Using the notation from the proof of Lemma~\ref{lemma_qdelta_bound_winning}, we have $x(S^+)=w(S^+)+\delta w(S^-)$ and 
  $\Vert w-x\Vert_1=2\delta w(S^-)$.   
    
  Generate $T$ by starting at $T=\emptyset$ and successively add a remaining player $i$ in $N\backslash T$ with maximal
  $x_i/w_i$, where all players $j$ with $w_j=0$ are taken in the first rounds, as long as $w(T)+w_i< q$. By construction $T$ is a losing coalition 
  of $[q;w]$ with $q-\Delta\le w(T)<q$, since the generating process did not stop earlier. 
  
  If $w(S^+)\ge q$, we have $T\subseteq S^+$ and $x(T)/w(T)\ge x(S^+)/w(S^+)=1+\frac{\delta w(S^-)}{w(S^+)}\ge 1+\delta w(S^-)$. Multiplying 
  by $w(T)$ and using $w(T)\ge q-\Delta$ yields
  \begin{equation*}
    x(T)\ge \left(1+\delta w(S^-)\right)w(T)\ge\left(1+\delta w(S^-)\right)(q-\Delta)= (q-\Delta) +\delta w(S^-)(q-\Delta). 
  \end{equation*}
  Since $x(T)\le q$, as $T$ is a losing coalition, we conclude $\delta w(S^-)\le\Delta/(q-\Delta)$, so that 
  $\Vert w-x\Vert_1<\frac{2\Delta}{q-\Delta}$.
  
  If $w(S^+)< q$, we have $S^+\subseteq T$, $x(T)=x(S^+)+x(T\backslash S^+)$, $w(T\backslash S^+)>0$, and $w(S^-)>0$. Since 
  $T\backslash S^+\subseteq S^-$, 
  $x(T\backslash S^+)/w(T\backslash S^+)\ge x(S^-)/w(S^-)$, so that
  \begin{eqnarray*}
    x(T)&=&x(S^+)+x(T\backslash S^+)\ge  w(S^+)+\delta w(S^-) +\frac{x(S^-)}{w(S^-)}\cdot\left(w(T)-w(S^+)\right)\\
    &\ge& w(S^+)+\delta w(S^-) +(1-\delta)\cdot\left(q-\Delta-w(S^+)\right)\\
    &=& \delta w(S^-)+q-\Delta-\delta q+\delta \Delta+\delta w(S^+)
    = q-\Delta+\delta(1-q+\Delta).
  \end{eqnarray*} 
  Since $x(T)\le q$, 
  $\delta\le \frac{\Delta}{1-q+\Delta}$, so that 
  $\Vert w-x\Vert_1\le \frac{2\Delta}{1-q+\Delta}$ due to $w(S^-)\le 1$.     
  
  So, for $q>\Delta$ we have $\Vert w-x\Vert_1 \le \frac{2\Delta}{\min\{q-\Delta,1-q+\Delta\}}\le \frac{2\Delta}{\min\{q-\Delta,1-q\}}$. 
  In order to show $\Vert w-x\Vert_1\le \frac{4\Delta}{\min\{q,1-q\}}$ it remains to consider the case $q\le 1-q$. For $q>2\Delta$, see 
  the start of the proof, we have $\Vert w-x\Vert_1\le \frac{2\Delta}{\min\{q-\Delta,1-q\}}\le \frac{2\Delta}{q-\Delta}\le 
  \frac{4\Delta}{q}\le \frac{4\Delta}{\min\{q,1-q\}}$.
\end{proof}


\end{document}